%% file: full_version.tex
\documentclass[11pt]{article}
\usepackage[utf8]{inputenc}
\usepackage{amsmath}
\usepackage{amsthm}
\usepackage{amsfonts}
\usepackage{natbib}
\usepackage{bm}
\usepackage{thm-restate}
\usepackage{multirow}
\usepackage{mathrsfs}

\usepackage{authblk}

\usepackage{graphicx}
\usepackage{subcaption}
\usepackage{float}

\usepackage{color}   
\usepackage{hyperref}
\hypersetup{
    linktoc=all,     
    linkcolor=blue,  
}

\usepackage[margin=1in]{geometry}

\usepackage{todonotes}

\usepackage[ruled,vlined]{algorithm2e}

\usepackage{color}   

\def\reals{{\mathbb R}}
\def\bw{\mathbf{w}}
\def\tbw{\tilde{\bw}}
\def\mE{\mathscr{E}}

\def\Cov{\mathrm{Cov}}
\def\Var{\mathrm{Var}}
\newcommand{\eps}{\epsilon}
\newcommand{\E}{\mathbb{E}}

\newcommand{\opt}{\mathrm{OPT}}
\newcommand{\R}{\mathbb{R}}

\usepackage{framed}
\usepackage{setspace}

\usepackage{hyperref}

\usepackage{todonotes}

\usepackage[ruled,vlined]{algorithm2e}

\newcommand{\ind}{\mathds{1}}

\def\calN{\mathcal{N}}

\def\calT{\mathcal{T}}

\def\calU{\mathcal{U}}

\def\bT{\mathcal{T}}
\def\reals{{\mathbb R}}
\def\Z{{\mathbb Z}}

\def\tw{\tilde{w}}

\newcommand{\N}{\mathbb{N}}
\newcommand{\cA}{\mathcal{A}}
\newcommand{\cB}{\mathcal{B}}
\newcommand{\cQ}{\mathcal{Q}}
\newcommand{\Ber}{\mathrm{Ber}}
\newcommand{\Beta}{\mathrm{Beta}}
\newcommand{\bq}{\mathbf{q}}

\def\barK{\bar{K}}
\def\barK{K} 

\theoremstyle{plain}
\newtheorem{theo}{Theorem}
 
\newtheorem{theorem}[theo]{Theorem}
\newtheorem{lemma}[theo]{Lemma}

\newtheorem{corollary}[theo]{Corollary}

\newtheorem{fact}[theo]{Fact}

\newtheorem{obs}[theo]{Observation}

\theoremstyle{definition}

\theoremstyle{remark}

\usepackage{dsfont}
\usepackage{cleveref}

\DeclareMathOperator*{\poly}{poly}

\DeclareMathOperator*{\Lap}{Lap}

\def\bw{\text{\textbf{w}}}
\def\Bern{\text{Bern}}

\SetKwInput{KwParam}{Parameters}
\SetKwInput{KwOracle}{Oracle Access}

\newcommand{\pasin}[1]{\textcolor{red}{[Pasin: #1]}}

\allowdisplaybreaks

\begin{document}

\title{
Private Rank Aggregation in Central and Local Models}
\author[1]{Daniel Alabi\footnote{Work done while an intern at Google Research.}}
\author[2]{Badih Ghazi}
\author[2]{Ravi Kumar}
\author[2]{Pasin Manurangsi}
\affil[1]{Harvard School of Engineering and Applied Sciences}
\affil[2]{Google Research, Mountain View, CA}
\affil[ ]{alabid@g.harvard.edu, \{badihghazi, ravi.k53\}@gmail.com, pasin@google.com}
\date{}

\maketitle

\begin{abstract}
In social choice theory, (Kemeny) rank aggregation is a well-studied
problem where the goal is to combine rankings from multiple voters
into a single ranking on the same set of items.
Since rankings can reveal preferences of voters (which a voter
might like to keep private), it is important to aggregate
preferences in such a way to preserve privacy.
In this work,
we present differentially private algorithms for
rank aggregation in the pure and approximate settings along with
distribution-independent utility upper and lower bounds.
In addition to bounds in the central model, we also present
utility bounds for the local model of differential privacy.
\end{abstract}

\clearpage

\section{Introduction}

The goal of rank aggregation is to find a central ranking based on rankings of two or more voters.  Rank aggregation is a basic formulation that is studied in  diverse disciplines ranging from social choice~\citep{arrow1963social}, voting~\citep{YL78, Young95}, and behavioral economics to machine learning~\citep{WistubaP20}, data mining~\citep{DworkKNS01}, recommendation systems~\citep{PennockHG00}, and information retrieval~\citep{MansouriZO21}.  The problem has a long and rich algorithmic history~\citep{ConitzerDK06, MeilaPPB07, Kenyon-MathieuS07, MandhaniM09, SchalekampZ09, AilonCN08}.

In many rank aggregation applications such as voting, it is imperative to find a central ranking such that the privacy of the voters who contributed the rankings is preserved.  A principled way to achieve this is to view the problem through the lens of the rigorous mathematical definition of differential privacy (DP)~\citep{DworkMNS06, DworkKMMN06}.  As in other privacy settings, it is then important to understand the trade-off between privacy and utility, i.e., the quality of the aggregation.

DP rank aggregation has been studied in a few recent papers~\citep{HayEM17, YanLL20, LiuLXZ20}.  While the algorithms in these works are guaranteed to be DP, their utility is provably good only under generative settings such as the Mallows model~\citep{Mallows57, FV86} or with other distributional assumptions.  The question we ask in this work is: can one circumvent such stylized settings and obtain a DP rank aggregation algorithm with provably good utility in the worst case?
We answer this affirmatively by providing a spectrum of
algorithms in the central and local DP models.

In the central model of DP, a trusted curator holds the 
\textit{non-private} exact rankings from
all users and aggregates them
into a single ranking while ensuring privacy. 
In the local model of DP,
this process would be done
without collating the \textit{exact} rankings of the users.
As an example,
consider a distributed
recommendation system that can be used to determine
a central ranking based on individual 
user rankings.
Suppose there are at least
four users in a recommendation system. Each user wishes to visit
a state with at least one Google office. In addition, each user
has their own criteria for deciding which office to visit.
The candidate Google offices are:
Mountain View ($M$), 
San Francisco ($S$),
New York ($N$),
Cambridge ($C$), and
Accra ($A$).
In Table~\ref{tab:ex}, we show each user's
ranking of the candidate offices based on different criteria.
Evidently, since each user wishes to visit a state or city for a
different reason, every user's ranking on the candidate
offices is also not the same. The goal of our work is to
aggregate user preferences into one single ranking on the
candidate offices while preserving the privacy of each user.

\begin{table}[b]
\begin{tabular}{| l | l | l |}\hline
User &
Criteria &
Ranking
\\\hline
1 &
Visit a relative in California &
$M<S<N<C<A$
\\\hline
2 &
Go on vacation at an African Safari &
$A<M<N<C<S$
\\\hline
3 &
Take pictures of the Statue of Liberty &
$N<C<S<M<A$
\\\hline
4 &
Touch the Android Lawn Statues &
$M<S<C<N<A$
\\\hline
\end{tabular}
\caption{Aggregating Preferences on Offices to Visit.
$X<Y$ denotes preference for $X$ over $Y$.}
\label{tab:ex}
\end{table}

\subsection{Background}

Let $[m] = \{1, \ldots, m\}$ be the universe of items to be ranked (we assume this is public information) and let $\mathbb{S}_m$ be the group of rankings (i.e., permutations) on $[m]$.  
For example, the universe of items could be the set of
all restaurants in New York.
Let $\Pi = \{\pi_1, \ldots, \pi_n\}$ be a given
set of rankings, where $n$ is the number of users and  each $\pi_k \in \mathbb{S}_m$ gives an ordering on the $m$ items.  We assume that the lower the position $\pi(j)$ of an item $j\in[m]$ in
a ranking $\pi$, the higher our preference for that item $j$.
For $i, j\in[m]$ such that $i\neq j$,
define $w_{ij}^{\Pi} =
\Pr_{\pi\sim\Pi} [\pi(i) < \pi(j)]$, i.e., $w_{ij}^{\Pi}$ is the fraction of rankings that rank item $i$ before $j$.
As noted by~\citet{AilonCN08}, for a graph on $m$ nodes, 
when the weight of each edge $(i, j)$ is $w_{ij}^{\Pi}$, rank aggregation reduces to the 
weighted version of the minimum Feedback Arc Set in Tournaments (FAST) problem.

Rank aggregation is based on the
\textit{Kendall tau} metric:
\[
K(\pi_1, \pi_2) = |\{(i, j)\,:\, \pi_1(i) < \pi_1(j)\text{ but } \pi_2(i) > \pi_2(j)\}|,
\]
for any two rankings $\pi_1, \pi_2\in \mathbb{S}_m$.
In other words, $K(\pi_1, \pi_2)$ is the number of pairwise
disagreements between the two permutations $\pi_1$ and $\pi_2$.
Note that $K(\pi, \pi) = 0$ for any permutation $\pi$ and the maximum value of $K(\cdot, \cdot)$ is
${m\choose 2}$. 
For example, if $\pi_1 = (1,2,3,4)$ and
$\pi_2 = (2,3,1,4)$, then
$K(\pi_1, \pi_2) = 2$.

We also define the \textit{average Kendall tau} distance
(or the \textit{Kemeny Score})
to a set $\Pi = \{\pi_1, \dots, \pi_n\}$ of rankings as
$\barK(\sigma, \Pi) = \frac{1}{n}\sum_{i=1}^nK(\sigma, \pi_i)$.
We use $\opt(\Pi)$ to denote $\min_{\sigma} \barK(\sigma, \Pi)$.  We say that a randomized algorithm $\cA$ obtains an 
\textit{$(\alpha, \beta)$-approximation} for the (Kemeny) rank aggregation problem
if, given $\Pi$, it outputs $\sigma$ such that $\E_\sigma[\barK(\sigma, \Pi)] \leq \alpha \cdot \opt(\Pi) + \beta$, where $\alpha$ is the  approximation ratio and $\beta$ is the additive error.  When $\beta = 0$, we call $\cA$ an \textit{$\alpha$-approximation} algorithm.

\paragraph{Differential Privacy.} We consider two sets $\Pi = \{\pi_1, \dots, \pi_n\}, \Pi' = \{\pi'_1, \dots, \pi'_n\}$ of rankings to be \emph{neighboring} if they differ on a single ranking. In this work, we consider both the \emph{central} and
\emph{local} models of DP. For $\epsilon > 0, \delta \geq 0$, a randomized algorithm $\cA: (\mathbb{S}_m)^{n} \rightarrow \mathbb{S}_m$ is \emph{$(\epsilon, \delta)$-differentially private (DP) in the central model} if for all neighboring sets $\Pi, \Pi'$ of rankings and for all $S \subseteq \mathbb{S}_m$, 
\[
\Pr[\cA(\Pi) \in S] \leq e^{\epsilon} \Pr[\cA(\Pi') \in S] + \delta.
\]
In other words, in the central model the algorithm $\cA$ can access all the input rankings and only the output aggregated ranking is required to be private.

In the (interactive) local model of DP~\citep{KasiviswanathanLNRS11}, the users retain their rankings and the algorithm $\cA$ is a randomized protocol between the users; let $\bT_{\cA}$ denote the transcript of the protocol\footnote{We refer the readers to~\citep{DuchiR19} for a more detailed formalization and discussion on interactivity in the local model.}. An algorithm $\cA$ is \emph{$(\epsilon, \delta)$-differentially private (DP) in the local model} if for all neighboring sets $\Pi, \Pi'$ of rankings and for all sets $S$ of transcripts, 
\[
\Pr[\bT_{\cA}(\Pi) \in S] \leq e^{\epsilon} \Pr[\bT_{\cA}(\Pi') \in S] + \delta.
\]

We say that $\cA$ is a pure-DP algorithm (denoted $\epsilon$-DP) when $\delta=0$ and is an approximate-DP algorithm otherwise.  

\subsection{Our Results}

We obtain polynomial-time pure- and approximate-DP approximation algorithms for rank aggregation in the central and local models.
Note that by using the exponential mechanism of~\cite{McSherryT07}, one can obtain an algorithm with approximation ratio of $1$ and additive error of $\tilde{O}(m^3/n)$; however this algorithm is computationally inefficient~\citep{HayEM17}.  

Our algorithms are based on two generic reductions. In our first reduction, we show that using standard DP mechanisms for aggregation (e.g., Laplace or Gaussian in the central model) to estimate the values of $w_{ij}$'s and then running an off-the-shelf approximation algorithm for rank aggregation (that is not necessarily private), preserves the approximation ratio and achieves additive errors of $O(m^4/n)$ for pure-DP in the central model, $O(m^3/n)$ for approximate-DP in the central model, and $O(m^3/\sqrt{n})$ in the local model. In the second reduction, we show how to improve the additive errors to $\tilde{O}(m^3/n)$, $\tilde{O}(m^{2.5}/n)$, and $\tilde{O}(m^{2.5}/\sqrt{n})$ respectively, while obtaining an approximation ratio of almost $5$. This reduction utilizes the query pattern properties of the KwikSort algorithm of~\citet{AilonCN08}, which only needs to look at $O(m \log m)$ entries $w_{ij}$'s. Roughly speaking, this means that we can add a smaller amount of noise per query, leading to an improved error bound. Our results are summarized in Tables~\ref{tab:comp} and~\ref{tab:localcomp}.

We remark that the idea of adding a smaller amount of noise when the algorithm uses fewer queries of $w_{ij}$'s was also suggested and empirically evaluated for the KwikSort algorithm by~\citet{HayEM17}. However, they did not prove any theoretical guarantee of the algorithm. Furthermore, our algorithm differs from theirs when KwikSort exceeds the prescribed number of queries: ours outputs the DP ranking computed using the higher-error algorithm that noises all $m^2$ entries whereas theirs just outputs the sorted ranking so far where the unsorted part is randomly ordered. The latter is insufficient to get the additive error that we achieve because leaving even a single pair unsorted can lead to error as large as $\Omega(1)$; even if this event happens with a small probability of $1/m^{O(1)}$, the error remains at $\Omega(1/m^{O(1)})$, which does not converge to zero as $n \to \infty$.

In addition to our algorithmic contributions, we also prove lower bounds on the additive error that roughly match the upper bounds for $n \geq \tilde{O}(m)$ in the case of pure-DP and for $n \geq \tilde{O}(\sqrt{m})$ for approximate-DP, even when the approximation ratio is allowed to be large.  Our lower bounds proceed by reducing from the $1$-way marginal problem and utilizing existing lower bounds for pure- and approximate-DP for the 
$1$-way marginals
problem~\citep{HardtT10,BunUV18,SteinkeU15,DworkSSUV15,SteinkeU17}.
\begin{table}[htb]
\small
\centering
\begin{tabular}{r | c | l}\hline
                     & $\alpha$ & $\beta$   \\\hline
 $\eps$-DP  & 
 \multirow{2}{*}{$5 + \xi$}  & $\frac{m^3\log m}{\eps n}$ \hfill (Corollary~\ref{cor:kwiksort-pure}) \\
 $(\eps, \delta)$-DP  & 
  & $\frac{m^{2.5}\sqrt{\log m}}{\eps n}\sqrt{\log \frac{1}{\delta}}$ \hfill (Corollary~\ref{cor:kwiksort-apx})  \\ \hline
 $\eps$-DP & 
 \multirow{2}{*}{$1 + \xi$} & $\frac{m^4}{\eps n}$ \hfill (Corollary~\ref{cor:PTAS-pure-DP}) 
 \\
  $(\eps, \delta)$-DP & 
  & $\frac{m^3}{\eps n}\sqrt{\log \frac{1}{\delta}}$ \hfill (Corollary~\ref{cor:PTAS-apx-DP})
 \\\hline
  $\eps$-DP  & 
 $1$ & $\frac{m^3}{\eps n}\log m$ \hfill (Hay et. al.,~\citeyear{HayEM17}) \\\hline
\end{tabular}
\caption{Guarantees of different $(\alpha, \beta)$-approximation algorithms in the \emph{central} model; here $\xi$ can be any positive constant.  We drop the big-O notation in $\beta$ for brevity.   All of our algorithms run in $\poly(nm)$ time, whereas the exponential mechanism (Hay et. al.,~\citeyear{HayEM17}) runs in $O(m! \cdot n)$ time.
\label{tab:comp}}
\end{table}
\begin{table}[htb]
\small
\centering
\begin{tabular}{r | c | l}\hline
                     & $\alpha$ & $\beta$   \\\hline
 $\eps$-DP  & 
 $1 + \xi$  & $\frac{m^3}{\eps\sqrt{n}}$ \hfill (\Cref{cor:PTAS-local-DP}) \\
 $\eps$-DP & 
 $5 + \xi$ & $\frac{m^{2.5}\log m}{\eps\sqrt{n}}$ \hfill (\Cref{cor:kwiksort-local})
 \\\hline
\end{tabular}
\caption{Guarantees of 
$(\alpha, \beta)$-approximation algorithms in the \emph{local} model.
As before, we drop the big-O notation in $\beta$. 
\label{tab:localcomp}}
\end{table}
\input{utility-draft}

For concreteness, we present the idea in Corollaries~\ref{cor:kwiksort-pure} and~\ref{cor:kwiksort-apx} as Algorithm~\ref{alg:dpks}.
We use a global counter $c$ to keep track of how many comparisons
have been done. Once $c > q$ (which happens with
negligible probability when we set $q = \Omega(m \log m)$), we default to using a PTAS with privacy budget of
$\eps/2$ (Corollary~\ref{cor:PTAS-pure-DP}).
If $\delta > 0$, we let $\calN$ be the Gaussian noise with standard deviation
$O_{\xi}\left(\frac{\sqrt{m \log m}}{\eps n}\sqrt{\log \frac{1}{\delta}}\right)$ and 
if $\delta = 0$, we let $\calN$ be drawn from $\Lap \left(0, O_{\xi}\left(\frac{m \log m}{\eps n}\right) \right)$.

\begin{algorithm}
{\bf Parameters:} {$\eps, \delta\in(0, 1]$, $q \in \N$} \\
\KwIn{$\Pi = \{\pi_1, \ldots, \pi_n\}, \calU \subseteq [m], c \leftarrow 0$}

\If {$\calU = \emptyset$} {
    \Return $\emptyset$
}

$\calU_L, \calU_R \leftarrow \emptyset$

Pick a pivot $i\in\calU$ uniformly at random

\For {$j\in\calU\setminus\{i\}$} {

$c \leftarrow c + 1$

\If {$c > q$} {
    Return PTAS with privacy budget of $\eps/2$ as final output \hfill $\triangleright$ See~\Cref{cor:PTAS-pure-DP}
}

$\tilde{w}_{ji} \gets w_{ji}^\Pi + \calN$
\hfill $\triangleright$ See text

\If {$\tilde{w}_{ji} > 0.5$} {
  Add $j$ to $\calU_L$
} \Else {
 Add $j$ to $\calU_R$
}

}

\Return \small DPKwikSort($\Pi$, $\calU_L, c$), $i$, DPKwikSort($\Pi$, $\calU_R, c$) \normalsize
\caption{DPKwikSort.}
\label{alg:dpks}
\end{algorithm}

\input{ldp}

\input{lb}

\section{Other Related Work}

In many disciplines (especially in social choice), 
rank aggregation appears in many different forms with
applications to
collaborative filtering and more general social computing~\citep{CohenSS99, PennockHG00, DworkKNS01}.
It has been shown previously by~\citet{arrow1963social} that
no voting rule (on at least three candidates) can satisfy
certain criteria at once. To circumvent such impossibility results,
we could rely on relaxed criteria such as the Condorcet
~\citep{YL78, Young95} condition, where we pick a candidate that
beats all others in head-to-head comparisons. The Kemeny ranking
~\citep{kemeny1962mathematical} can be used to obtain a
Condorcet winner (if one exists) and also rank candidates in such
a way as to minimize disagreements between the rankings from the voters.
The Kemeny ranking problem is NP-hard even for four voters~\citep{BTT89, CohenSS99, DworkKNS01}. 
To overcome this, computationally efficient approximation algorithms 
have been devised~\citep{borda, DG77, ConitzerDK06, Kenyon-MathieuS07, AilonCN08}. Our results rely on such algorithms.

In correlation clustering~\citep{BansalBC04}, we are given a graph whose edges are labeled green or red. The goal is to find a clustering that minimizes the
number of pairwise disagreements with the input graph, i.e., the number of intra-cluster red
edges and inter-cluster green edges. 
%
Correlation clustering is closely related to ranking problems (see e.g.,~\citep{AilonCN08}).
\citet{Bun0K21} recently tackled correlation clustering under DP constraints, although their notion of neighboring datsets---roughly corresponding to changing an edge---is very different than ours.

\section{Conclusion \& Future Work}

In this work,
we have provided several DP algorithms and lower bounds for the rank
aggregation problem in the central and local models.
Since each of our algorithms achieves either the near-optimal approximation ratio or the near-optimal additive error (but not both), an immediate
open question here is if one can get the best of both worlds.

Furthermore, recall that our local DP algorithm in \Cref{cor:kwiksort-local} requires interactivity, which seems inherent for any QuickSort-style algorithms.  It is interesting if one can achieve similar guarantees with a \emph{non-interactive} local DP algorithm. We point out that separations between interactive and non-interactive local models are known (see, e.g.,~\citep{DaganF20} and the references therein); thus, it is possible that such a separation exists for rank aggregation.
Lastly, it is interesting to see if we can extend our results to other related problems---such as consensus
and correlation clustering---that rely
on the (weighted) minimum FAST problem for their non-private approximation algorithms.


\newpage

\bibliographystyle{plainnat}
\bibliography{main}

\newpage

\onecolumn 

\appendix
{\Huge \bf Supplementary Material}

\input{sm-local}
\input{sm-lb}

\end{document}

%% file: utility-draft.tex
\section{Notation}

Let $\bw^{\Pi}$ denote the $[m]\times[m]$ matrix whose $(i,j)$th entry is $w_{ij}^{\Pi}$.  Note that
the average Kendall tau distance can also be written as
\begin{align*}
\barK(\sigma, \Pi) = \sum_{i, j \in [m]}
\ind[\sigma(i) \leq \sigma(j)] \cdot w_{ji}^{\Pi},
\end{align*}
where $\ind[\cdot]$ is the binary indicator function.  Hence, we may write $\barK(\sigma, \bw^{\Pi})$ to denote $\barK(\sigma, \Pi)$ and $\opt(\bw^{\Pi})$ to denote $\opt(\Pi)$.
When $\sigma$ or $\Pi$ are clear from the context, we may drop them for brevity.

To the best of our knowledge, all known rank aggregation algorithms only use $\bw$ and do not directly require the input rankings $\Pi$. The only assumption that we will make throughout is that the $w_{ij}$'s satisfy the \emph{probability constraint}~\citep{AilonCN08}, which means that $w_{ij} \geq 0$ and $w_{ij} + w_{ji} = 1$ for all $i, j \in [m]$ where
$i\neq j$. Again, the non-private algorithms of~\citep{AilonCN08,Kenyon-MathieuS07} that we will use in this paper obtained approximation ratios under this assumption.

In fact, many algorithms do not even need to look at all of $\bw$. Due to this, we focus on algorithms that allow \emph{query} access to $\bw$. Besides these queries, the algorithms do not access the input otherwise. Our generic reductions below will be stated for such algorithms.

\section{Algorithms in the Central Model}

In this section we present pure- and approximate-DP algorithms for rank aggregation in the central model of DP.  Our algorithms follow from two generic reductions: for noising queries and for reducing the additive error.

\paragraph{Noising Queries.}
To achieve DP, we will have to add noise to the query answers. In this regard, we say that a query answering algorithm $\cQ$ incurs \emph{expected error} $e$ if for every $i, j$ the answer $\tw_{ij}$ returned by $\cQ$ satisfies $\E[|\tw_{ij} - w_{ij}|] \leq e$. Furthermore, we assume that $\cQ$ computes the estimate $\tw_{ij}$ for all $i, j \in [m]$ (non-adaptively) but only a subset of $\tw_{ij}$ is revealed when queried by the ranking algorithm; let $\tbw$ denote the resulting matrix of $\tw_{ij}$'s.  In this context, we say that $\cQ$ satisfies $(\eps, \delta)$-DP for $q$ (adaptive) queries if $\cQ$ can answer $q$ such $\tw_{ij}$'s while respecting $(\eps, \delta)$-DP.  In our algorithms, we only use $\cQ$ that adds independent Laplace or Gaussian noise to each $w_{ij}$ to get $\tw_{ij}$ but we state our reductions in full generality as they might be useful in the future.

\subsection{Reduction I: Noising All Entries}
\label{sec:reduction1}

We first give a generic reduction from a \emph{not necessarily private} algorithm to DP algorithms. A simple form is:

\begin{theorem} \label{thm:red-wpone}
Let $\alpha > 1, e, \eps > 0, q \in \N$, and $\delta \in [0, 1]$. Suppose that there exists a polynomial-time (not necessarily private) $\alpha$-approximation algorithm $\cA$ for the rank aggregation problem that always makes at most $q$ queries. Furthermore, suppose that there exists a polynomial-time query answering algorithm $\cQ$ with expected error $e$ that is $(\eps, \delta)$-DP for answering at most $q$ queries. 

Then, there exists a polynomial-time $(\eps, \delta)$-DP $(\alpha, (\alpha + 1)m^2 e)$-approximation algorithm $\cB$ for rank aggregation.
\end{theorem}

This follows from an easy fact below that the error in the cost is at most the total error from querying all pairs of $i, j$.

\begin{fact} \label{fact:err-weight-to-err-cost}
For any $\tbw, \bw$ and $\sigma \in \mathbb{S}_m$, we have $|\barK(\sigma, \bw) - \barK(\sigma, \tbw)| \leq \sum_{i, j \in [m]} |w_{ji} - \tw_{ji}|$.
\end{fact}

\begin{proof}[Proof of Theorem~\ref{thm:red-wpone}]
Our algorithm $\cB$ simply works by running $\cA$, and every time $\cA$ queries for a pair $i, j$, it returns the answer using the algorithm $\cQ$. The output ranking $\sigma$ is simply the output from $\cA$. Since only $\cQ$ is accessing the input directly and from our assumption, it immediately follows that $\cB$ is $(\eps, \delta)$-DP.

Since we may view $\cA$ as having the input instance $\tbw$ (obtained by querying $\cQ$), the $\alpha$-approximation guarantee of $\cA$ implies that
\begin{align} \label{eq:apx-1}
\E_{\sigma}[\barK(\sigma, \tbw)] \leq \alpha \cdot \opt(\tbw).
\end{align}
By applying Fact~\ref{fact:err-weight-to-err-cost} twice, we then arrive at
\begin{align*}
&\E_{\sigma, \tbw}[\barK(\sigma, \bw)] \\
(\text{Fact}~\ref{fact:err-weight-to-err-cost}) &\leq m^2e + \E_{\sigma, \tbw}[\barK(\sigma, \tbw)] \\
\eqref{eq:apx-1} &\leq m^2e + \alpha \cdot \opt(\tbw) \\
(\text{Fact}~\ref{fact:err-weight-to-err-cost}) &\leq m^2e + \alpha \cdot (\opt(\bw) + m^2e) \\
&= \alpha \cdot \opt(\bw) + (\alpha + 1)m^2 e. \qedhere
\end{align*}
\end{proof}


The above reduction itself can already be applied to several algorithms, although it does not yet give the optimal additive error. As an example, we may use the (non-private) PTAS of~\citet{Kenyon-MathieuS07}; the PTAS requires \emph{all} $w_{ij}$'s meaning that $q \leq m^2$ and thus we may let $\cQ$ be the algorithm that adds $\Lap\left(0, \frac{m^2}{\eps n}\right)$ noise to each query\footnote{More precisely, to satisfy the probability constraints, we actually add the noise to each $w_{ij}$ only for all $i < j$ and clip it to be between 0 and 1 to arrive at $\tw_{ij}$. We then let $\tw_{ji} = 1 - \tw_{ij}$.}. This immediately implies the following:
\begin{corollary} 
\label{cor:PTAS-pure-DP}
For any $\xi, \eps > 0$, there exists an $\eps$-DP $\left(1 + \xi, O_{\xi}\left(\frac{m^4}{\eps n}\right)\right)$-approximation algorithm for the rank aggregation problem in the central model.
\end{corollary}

Similarly, if we instead add Gaussian noise with standard deviation $O\left(\frac{m}{\eps n}\sqrt{\log \frac{1}{\delta}}\right)$ to each query, we arrive at the following result:
\begin{corollary} \label{cor:PTAS-apx-DP}
For any $\xi, \eps, \delta > 0$, there exists an $(\eps, \delta)$-DP \\ $\left(1 + \xi, O_{\xi}\left(\frac{m^3}{\eps n}\sqrt{\log \frac{1}{\delta}}\right)\right)$-approximation algorithm for the rank aggregation problem in the central model.
\end{corollary}


\subsection{Reduction II: Improving the Additive Error}
\label{sec:reduction2}

A drawback of the reduction in Theorem~\ref{thm:red-wpone} is that it requires the algorithm to always make at most $q$ queries. This is a fairly strong condition and it cannot be applied, e.g., to QuickSort-based algorithms that we will discuss below. We thus give another reduction that works even when the algorithm makes at most $q$ queries \emph{with high probability}. The specific guarantees are given below.

\begin{theorem} \label{thm:red-expectation}
Let $\alpha, e, \eps > 0, q \in \N$, and $\zeta, \delta \in [0, 1]$. Suppose that there exists a polynomial-time (not necessarily private) $\alpha$-approximation algorithm $\cA$ for rank aggregation that, with probability $1 - \frac{\zeta}{m^4}$, makes at most $q$ queries. Furthermore, suppose that there exists a polynomial-time query answering algorithm $\cQ$ with expected error $e$ that is $(\eps/2, \delta)$-DP for answering at most $q$ queries. 

Then, there exists a polynomial-time $(\eps, \delta)$-DP $\left(\alpha + \zeta, \beta\right)$-approximation  algorithm $\cB$ for the rank aggregation problem where
$\beta = O_{\alpha}\left(m^2 e + \frac{1}{\eps n}\right)$.
\end{theorem}

\begin{proof}[Proof of Theorem~\ref{thm:red-expectation}]
Our algorithm $\cB$ simply works by running $\cA$, and every time $\cA$ queries for a pair $i, j$, it returns the answer using the algorithm $\cQ$. If $\cA$ ends up using at most $q$ queries, then it returns the output $\sigma^\cA$ of $\cA$. Otherwise, it runs the $\eps/2$-DP algorithm from Corollary~\ref{cor:PTAS-pure-DP} with $\xi = 1$ and return its output $\sigma^*$. The $(\eps, \delta)$-DP guarantee of the algorithm is immediate from the assumptions and the basic composition property of DP (see, e.g., \citet{DworkR14}).

We next analyze the expected Kendall tau distance of the output. To do this, let $\mE_{q}$ denote the event that $\cA$ uses more than $q$ queries. From our assumption, we have $\Pr[\mE_{q}] \leq \zeta/m^4$. Furthermore, observe that
\begin{align}
\E_{\sigma^\cA}[\barK(\sigma^\cA, \tbw)] 
&\geq \Pr[\mE_q] \cdot \E_{\sigma^\cA}[\barK(\sigma^\cA, \tbw) \mid \mE_q]. \label{eq:expected-cost-conditioned-on-small-num-queries}
\end{align}
Now, let $\sigma$ denote the output of our algorithm $\cB$. We have
\begin{align*}
&\E_{\sigma, \tbw}[\barK(\sigma, \tbw))] \\
&\leq \Pr[\mE_q] \cdot \E_{\sigma^\cA, \tw}[\barK(\sigma^\cA, \bw) \mid \mE_q] 
+ \Pr[\neg \mE_q] \cdot \E_{\sigma^*}[\barK(\sigma^*, \bw)] \\
& \leq \E_{\sigma^\cA}[\barK(\sigma^\cA, \tbw)] 
+ \frac{\zeta}{m^4} \cdot \left(2 \cdot \opt(\bw) + O_{\alpha}\left(\frac{m^4}{\eps n}\right) \right) \\
&\leq \alpha \cdot \opt(\bw) + O_{\alpha}(m^2 e) 
+ \frac{2\zeta}{m^4} \cdot \opt(\bw) + O_{\alpha}\left(\frac{1}{\eps n}\right) \\
&\leq (\alpha + \zeta) \cdot \opt(\bw) + O_{\alpha}\left(m^2 e + \frac{1}{\eps n}\right),
\end{align*}
where the second inequality follows from~\eqref{eq:expected-cost-conditioned-on-small-num-queries} and the approximation guarantee of the algorithm from Corollary~\ref{cor:PTAS-pure-DP}, the third inequality follows from similar computations as in the proof of Theorem~\ref{thm:red-wpone}, and the last inequality follows because we may assume w.l.o.g. that $m \geq 2$.
\end{proof}


\subsubsection{Concrete Bounds via KwikSort.}

\citet{AilonCN08} devise an algorithm \emph{KwikSort}, inspired by QuickSort, to determine a ranking on $m$ candidates and they show that this algorithm yields a 5-approximation for rank aggregation. It is known in the classic sorting literature that with high probability, QuickSort
(and by extension, KwikSort) performs $O(m\log m)$ comparisons
on $m$ candidates to order these $m$ items. Specifically, with probability at least $1 - \xi / m^4$, the number of comparisons between items would be at most $q = O_{\xi}(m \log m)$. Plugging this into Theorem~\ref{thm:red-expectation} where $\cQ$ adds $\Lap\left(0, \frac{q}{\eps n}\right)$ noise to each query, we get:
\begin{corollary} \label{cor:kwiksort-pure}
For any $\xi, \eps > 0$, there is an $\eps$-DP $\left(5 + \xi, O_{\xi}\left(\frac{m^3 \log m}{\eps n}\right)\right)$-approximation algorithm for the rank aggregation problem in the central model.
\end{corollary}

If we use Gaussian noise with standard deviation $O\left(\frac{\sqrt{q}}{\eps n}\sqrt{\log \frac{1}{\delta}}\right)$ instead of the Laplace noise, then we get:
\begin{corollary} \label{cor:kwiksort-apx}
For any $\xi, \eps, \delta > 0$, there exists an $(\eps, \delta)$-DP \\ $\left(5 + \xi, O_{\xi}\left(\frac{m^{2.5}\sqrt{\log m}}{\eps n}\sqrt{\log \frac{1}{\delta}}\right)\right)$-approximation algorithm for the rank aggregation problem in the central model.
\label{lem:gaussian}
\end{corollary}


Although the approximation ratios are now a larger constant (arbitrarily close to 5), the above two corollaries improve upon the additive errors in  Corollaries~\ref{cor:PTAS-pure-DP} and~\ref{cor:PTAS-apx-DP} by a factor of $\tilde{\Theta}(m)$ and $\tilde{\Theta}(\sqrt{m})$ respectively.

%% file: ldp.tex
\newcommand{\ts}{\tilde{s}}
\newcommand{\bv}{\mathbf{v}}
\newcommand{\bs}{\mathbf{s}}
\newcommand{\tbs}{\tilde{\bs}}
\newcommand{\cP}{\mathcal{P}}

\section{Algorithms in the Local Model}


We next consider the local model~\citep{KasiviswanathanLNRS11}, which is more stringent than the central model. We only focus on pure-DP algorithms in the local model since there is a strong evidence\footnote{\citet{BunNS19,JosephMN019} give a generic transformation from any \emph{sequentially interactive} local approximate-DP protocols to a pure-DP one while retaining the utility. However, this does not apply to the full interaction setting.} that approximate-DP does not help improve utility in the local model~\citep{BunNS19,JosephMN019}.

\subsection{Reduction I: Noising All Entries}

Similar to our algorithms in the central model, we start with a simple reduction that works for any non-private ranking algorithm by noising each $w_{ji}$. This is summarized below. 

\begin{theorem} \label{thm:red-wpone-local}
Let $\alpha > 1, \eps > 0$, and $\delta \in [0, 1]$. Suppose that there exists a polynomial-time (not necessarily private) $\alpha$-approximation algorithm $\cA$ for the rank aggregation problem.  Then, there exists a polynomial-time $\eps$-DP $\left(\alpha, O_{\alpha}\left(\frac{m^3}{\eps \sqrt{n}}\right)\right)$-approximation algorithm $\cB$ for the rank aggregation problem in the local model.
\end{theorem}

The proof of \Cref{thm:red-wpone-local} is essentially the same as that of \Cref{thm:red-wpone} in the central model, except that we use the query answering algorithm of~\citet{duchi2014local}, which can answer all the $m^2$ entries with expected error of $O\left(\frac{m}{\eps \sqrt{n}}\right)$ per entry; this indeed results in $O\left(\frac{m^3}{\eps \sqrt{n}}\right)$ additive error in the above theorem. We remark that if one were to naively use the randomized response algorithm on each entry, the resulting error would be $O\left(\frac{m^2}{\eps \sqrt{n}}\right)$ per entry; in other words,~\citet{duchi2014local} saves a factor of $\Theta(m)$ in the error. The full proof of \Cref{thm:red-wpone-local} is deferred to the Supplementary Material (SM).

Plugging the PTAS of~\citet{Kenyon-MathieuS07} into \Cref{thm:red-wpone-local}, we arrive at the following: 
\begin{corollary} 
\label{cor:PTAS-local-DP}
For any $\xi, \eps > 0$, there exists an $\eps$-local DP $\left(1 + \xi, O_{\xi}\left(\frac{m^3}{\eps\sqrt{n}}\right)\right)$-approximation algorithm for the rank aggregation problem in the local model.
\end{corollary}

\subsection{Reduction II: Improving the Additive Error}

Similar to the algorithm in the central model, intuitively, it should be possible to reveal only the required queries while adding a smaller amount of noise. Formalizing this, however, is more challenging because the algorithm of~\citet{duchi2014local} that we use is \emph{non-interactive} in nature, meaning that it reveals the information about the entire vector at once, in contrast to the Laplace or Gaussian mechanisms for which we can add noise and reveal each entry as the algorithm progresses. Fortunately,~\citet{Bassily19} provides an algorithm that has accuracy similar to~\citet{duchi2014local} but also works in the adaptive setting. However, even Bassily's algorithm does not directly work in our setting yet: the answers of latter queries depend on that of previous queries and thus we cannot define $\tbw$ directly as we did in the proof of~\Cref{thm:red-expectation}. 

\subsubsection{Modifying Bassily's Algorithm.}

We start by modifying the algorithm of~\citet{Bassily19} so that it fits in the ``query answering algorithm'' definition we used earlier in the central model. 
However, while the previous algorithms can always support up to $q$ queries, the new algorithm will only be able to do so with high probability. That is, it is allowed to output $\perp$ with a small probability. This is formalized below.

\begin{lemma} \label{lem:bassily}
For any $\eps, \zeta > 0$ and $q > 10 m \log m$, there exists an $\eps$-DP query-answering algorithm $\cQ$ in the local model such that
\begin{itemize}
\item If we consider $\perp$ as having zero error, its expected error per query is at most $O\left(\frac{q}{\eps\sqrt{nm}}\right)$.
\item For any (possibly adaptive) sequence $i_1j_1, \dots, i_{q}j_{q}$ of queries, the probability that $\cQ$ outputs $\perp$ on any of the queries is at most $\exp(-0.1q/m)$.
\end{itemize}
\end{lemma}

For simplicity of the presentation, we assume that $n$ is divisible by $m$. We remark that this is without loss of generality as otherwise we can imagine having ``dummy'' $t < m$ users so that $n + t$ is divisible by $m$, where these dummy users do not contribute to the queries.

The algorithm is simple: it randomly partitions the $n$ users into sets $\cP_1, \dots, \cP_m$. Then, on each query $ji$, it uses a \emph{random} set to answer the query (via the \emph{Randomized Response} mechanism~\citep{Warner65}) with $\eps_0$-DP where $\eps_0 = 0.5\eps m / q$. Finally, the algorithm outputs $\perp$ if that particular set has already been used at least $2q/m$ times previously.

We remark that the main difference between the original algorithm of~\citet{Bassily19} and our version is that the former partitions the users into $q$ sets and use the $t$th set to answer the $t$th query. This unfortunately means that an answer to a query $ji$ may depend on which order it was asked, which renders our utility analysis of DP ranking algorithms invalid because $\tbw$ is not well defined. Our modification overcomes this issue since the partition used for each $ji$ (denoted by $\ell(j, i)$ below) is independently chosen among the $m$ sets.

\begin{proof}[Proof of \Cref{lem:bassily}]
The full algorithm is described in Algorithm~\ref{alg:query-local}; here $(\cP_1, \dots, \cP_m)$ is a random partition where each set consists of $n/m$ users and $\ell(j, i)$ is i.i.d. uniformly at random from $[m]$. Moreover, $c$ is a vector of counters, where $c(p)$ is initialized to zero for all $p \in [m]$.

\begin{algorithm}
{\bf Parameters:} {$\eps > 0$, $q \in \N$, Partition $(\cP_1, \dots, \cP_m)$ of $[n]$, $\ell: [m] \times [m] \to [m]$} \\
\KwIn{$\Pi = \{\pi_1, \ldots, \pi_n\}, c: [m] \to \Z$}
	
	$\eps_0 \leftarrow 0.5\eps m / q$

	$d_\eps\leftarrow\frac{e^{\eps_0} + 1}{e^{\eps_0} - 1}$

    \For {user $k \in \cP_{\ell(j, i)}$} {
        $c(\ell(j, i)) \leftarrow c(\ell(j, i)) + 1$
        
        \If {$c(\ell(j, i)) > 2q/m$} {
            \Return $\perp$
        }

        $w_{ji}^k = \ind[\pi_k(j) < \pi_k(i)]$
        
        Let $\tilde{w}^k_{ji} \gets$ 
        $\begin{cases}
        d_\eps & \text{ w.p. } 1/2\cdot(1+w_{ji}^k/ d_\eps) \\
        -d_\eps & \text{ w.p. } 1/2\cdot(1-w_{ji}^k/ d_\eps)
        \end{cases}$
}

\Return $\tw_{ji} := \frac{m}{n} \cdot \sum_{k \in \cP_{\ell(j, i)}} \tw^k_{ji}$

\caption{$\eps$-DP Adaptive Query Answering Algorithm in the Local Model.}
\label{alg:query-local}
\end{algorithm}

We will now prove the algorithm's privacy guarantee. Let us consider a user $k$; suppose that this user belongs to $\cP_\ell$. By definition of the algorithm, this user applies the Randomized Response algorithm at most $2q/m$ times throughout the entire run. Since each application is $\eps_0$-DP (see, e.g.,~\citep{KasiviswanathanLNRS11}), basic composition of DP ensures that the entire algorithm is $(2q/m) \cdot \eps_0 = \eps$-DP.

We next analyze its accuracy guarantee. Fix a query $ji$. For every user $k \in [n]$, let $B_k$ denote $\ind[k \in \cP_{\ell(j, i)}]$. Notice that $\tw_{ji} = \sum_{k \in [n]} \frac{m}{n} \cdot B_k \cdot \tw^k_{ji}$ and $\E[B_k \cdot \tw^k_{ji}] = \frac{1}{m} \cdot w^k_{ji}$. Thus, we have $\E[\tw_{ji}] = w_{ji}$. Furthermore, we have
\begin{align}
\Var(\tw_{ji}) = \left(\frac{m}{n}\right)^2 \cdot \Bigg(\sum_{k \in [n]} \Var(B_k \cdot \tw^k_{ji}) \nonumber \\ \qquad + \sum_{1 \leq k < k' \leq n} \Cov(B_k \cdot \tw^k_{ji}, B_{k'} \cdot \tw^{k'}_{ji}) \Bigg). \label{eq:var-calculation}
\end{align}

Observe that $\Var(B_k \cdot \tw^k_{ji}) \leq \E[(B_k \cdot \tw^k_{ji})^2] = d^2_\eps / m = O\left(\frac{1}{m\eps_0^2}\right) = O\left(\frac{q^2}{m^3\eps^2}\right)$. Moreover, since $\cP_{\ell}$ is a random subset of $n/m$ users, $B_k$ and $B_{k'}$ are negatively correlated, meaning that $\Cov(B_k \cdot \tw^k_{ji}, B_{k'} \cdot \tw^{k'}_{ji}) \leq 0$. Plugging these back into~\eqref{eq:var-calculation} yields
$\Var(\tw_{ji}) \leq O\left(\frac{q^2}{n m \eps^2}\right)$.
This, together with the fact that $\tw_{ji}$ is an unbiased estimator of $w_{ji}$, implies that $\E[|\tw_{ji} - w_{ji}|] \leq O\left(\frac{q}{\eps\sqrt{nm}}\right)$ as desired.

Finally, we bound the probability that the algorithm outputs $\perp$. Since the $\ell(j, i)$'s are i.i.d. drawn from $[m]$, we can apply the Chernoff bound and the union bound (over the $m$ partitions) to conclude that the probability that any sequence of $q$ queries will trigger the algorithm to return $\perp$ is at most $m \cdot \exp(-q/(3m)) \leq \exp(-0.1q/m)$, where the inequality follows from our assumption that $q \geq 10 m \log m$.
\end{proof}

\subsubsection{The Reduction.}

Having described and analyzed the modified version of Bassily's algorithm, we can now describe the main properties of the second reduction in the local model:

\begin{theorem} \label{thm:red-expectation-local}
Let $\alpha, \eps > 0, q \in \N$, and $\zeta, \delta \in [0, 1]$ such that $q \geq 10 m \log(m/\zeta)$. Suppose that there exists a polynomial-time (not necessarily private) $\alpha$-approximation algorithm $\cA$ for the rank aggregation problem that with probability $1 - \frac{\zeta}{m^4}$ makes at most $q$ queries. 

Then, there is a polynomial-time $\eps$-DP $\left(\alpha + \zeta, \beta\right)$-approximation  algorithm $\cB$ for the rank aggregation problem where
$\beta = O_{\alpha}\left(\frac{m^{1.5} q}{\eps\sqrt{n}}\right)$ in the local model.
\end{theorem}

The proof of \Cref{thm:red-expectation-local} follows its counterpart in the central model (\Cref{thm:red-expectation}); the main difference is that, instead of stopping after $q$ queries previously, we stop upon seeing $\perp$ (in which case we run the PTAS from~\Cref{cor:PTAS-local-DP}). The full proof is deferred to SM.

As in the central model, if we apply the concrete bounds for the KwikSort algorithm (which yields 5-approximation and requires $O(m \log m)$ queries w.h.p.) to \Cref{thm:red-expectation-local}, we arrive at the following corollary: 

\begin{corollary} \label{cor:kwiksort-local}
For any $\xi, \eps > 0$, there is an $\eps$-DP $\left(5 + \xi, O_{\xi}\left(\frac{m^{2.5} \log m}{\eps \sqrt{n}}\right)\right)$-approximation algorithm for the rank aggregation problem in the local model.
\end{corollary}

For concreteness, we also provide the full description in Algorithm~\ref{alg:adaptiveldpks}. Again, similar to the setting of Algorithm~\ref{alg:query-local}, we pick $(\cP_1, \dots, \cP_m)$ to be a random partition where each set consists of $n/m$ users and $\ell(j, i)$ i.i.d. uniformly at random from $[m]$, and we initialize $c(p) = 0$ for all $p \in [m]$. Here $\tau$ is the threshold that is set to be $\Omega(\log m)$.

\begin{algorithm}
{\bf Parameters:} {$\eps > 0$, $\tau > 0$, Partition $(\cP_1, \dots, \cP_m)$ of $[n]$, $\ell \leftarrow [m] \times [m] \to [m]$} \\
\KwIn{$\Pi = \{\pi_1, \ldots, \pi_n\}, \calU \subseteq [m], c: [m] \to \Z$}

\If {$\calU = \emptyset$} {
    \Return $\emptyset$
}

$\calU_L, \calU_R \leftarrow \emptyset$

Pick a pivot $i\in\calU$ uniformly at random

\For {$j\in\calU\setminus\{i\}$} {

    \For {user $k \in \cP_{\ell(j, i)}$} {
        $c(\ell(j, i)) \leftarrow c(\ell(j, i)) + 1$
        
        \If {$c(\ell(j, i)) > \tau$} {
            Run (local) PTAS with privacy budget of $\eps/2$ \hfill $\triangleright$ See~\Cref{cor:PTAS-local-DP}
        }
        
        $\tilde{w}^k_{ji} \gets$ 
        $\begin{cases}
        d_\eps & \text{ w.p. } 1/2\cdot(1+w_{ji}^k/ d_\eps) \\
        -d_\eps & \text{ w.p. } 1/2\cdot(1-w_{ji}^k/ d_\eps)
        \end{cases}$
    }
    
    $\tw_{ji} \leftarrow \frac{m}{n} \cdot \sum_{k \in \cP_{\ell(j, i)}} \tw^k_{ji}$

    \If {$\tilde{w}_{ji} > 1/2$} {
        Add $j$ to $\calU_L$
    } \Else {
        Add $j$ to $\calU_R$
    }
    
}

\Return LDPKwikSort($\Pi$, $\calU_L, c$), $i$, LDPKwikSort($\Pi$, $\calU_R, c$)

\caption{LDPKwikSort.
}
\label{alg:adaptiveldpks}
\end{algorithm}

%% file: lb.tex
\section{Lower Bounds}
\label{sec:lower}

To describe our lower bounds, recall that the ``trivial'' additive error that is achieved by outputting an \emph{arbitrary} ranking is $O(m^2)$. In this sense, our algorithms achieve ``non-trivial'' additive error when $n \geq \tilde{O}(m / \eps)$ for pure-DP (\Cref{cor:kwiksort-pure}) and $n \geq \tilde{O}(\sqrt{m} / \eps)$ for approximate-DP (\Cref{cor:kwiksort-apx}) in the central model and when $n \geq \tilde{O}(m / \eps^2)$ in the local model (\Cref{cor:kwiksort-local}).  In this section we show that this is essentially the best possible, even when the multiplicative approximation ratio is allowed to be large. Specifically, we show the following results for pure-DP and approximate-DP respectively in the \emph{central} model.

\begin{theorem} \label{thm:lb-pure-dp}
For any $\alpha, \eps > 0$, there is no $\eps$-DP $(\alpha, 0.01 m^2)$-approximation algorithm for rank aggregation in the central model  for $n = o(m / \eps)$.
\end{theorem}

\begin{theorem} \label{thm:lb-apx-DP-all-eps}
For any constant $\alpha > 0$ and any $\eps \in (0, 1]$, there exists $c > 0$ (depending on $\alpha$) such that there is no $(1, o(1/n))$-DP $(\alpha, c m^2)$-approximation algorithm for rank aggregation in the central model for $n = o(\sqrt{m} / \eps)$.
\end{theorem}

We stress that any lower bound in the central model also applies to DP algorithms in the local model. Specifically, the sample complexity in \Cref{thm:lb-pure-dp} also matches that in~\Cref{cor:kwiksort-local} to within a factor of $O(1/\eps)$.

Due to space constraints, we will only describe high-level ideas here. The full proofs are deferred to SM.

\paragraph{Proof Overview: Connection to (1-Way) Marginals.} 
We will describe the intuition behind the proofs of Theorems~\ref{thm:lb-pure-dp} and~\ref{thm:lb-apx-DP-all-eps}. At the heart of the proofs, we essentially reduce from the 1-way marginal problem. For $d, t \in \N$, let $x \mapsto \pi_{d, t}^x$ denote the mapping from $\{-1, +1\}^d$ to $\mathbb{S}_{2d + t}$ defined by $\pi_{d, t}^x(\ell) = \ell$ for all $\ell \in \{d + 1, \dots, d + t\}$ and
\begin{align*}
\pi_{d, t}^x(j) :=
\begin{cases}
j &\text{ if } x_j = 1 \\
j + d + t &\text{ if } x_j = -1,
\end{cases}
\end{align*}
\begin{align*}
\pi_{d, t}^x(j + d + t) :=
\begin{cases}
j + d + t &\text{ if } x_j = 1 \\
j &\text{ if } x_j = -1,
\end{cases}
\end{align*}
for all $j \in [d]$.
In words, $\pi_{d, t}^x$ is an ordering where we start from the identity and switch the positions of $j$ and $j + d + t$ if $x_j = +1$.

Recall that in the \emph{(1-way) marginal problem}, we are given vectors $x^1, \dots, x^n \in \{-1, +1\}^d$ and the goal is to determine $\frac{1}{n} \sum_{i \in [n]} x^i$. The connection between this problem and the rank aggregation problem is that a ``good'' aggregated rank on $\pi_{d, t}^{x^1}, \dots, \pi_{d, t}^{x^n}$ must ``correspond'' to a $d$-dimensional vector whose sign is similar to the 1-way marginal. To formalize this, let us define the ``inverse'' operation of $x \mapsto \pi_{d, t}^x$, which we denote by $\rho_{d, t}: \mathbb{S}_{2d + t} \to \{-1, +1\}^d$ as follows:
\begin{align*}
\rho_{d, t}(\pi)_j = 
\begin{cases}
-1 \text{ if } \pi(j) < \pi(j + d + t) \\
+1 \text{ otherwise}.
\end{cases}
\end{align*}
The aforementioned observation can now be formalized:
\begin{obs} \label{obs:marginal-to-ranking}
For any $x^1, \dots, x^n \in \{-1, +1\}^d$ and any $\sigma \in \mathbb{S}_{2d}$, we have
\begin{align*}
&\barK(\sigma, \{\pi_{d, t}^{x^1}, \dots, \pi_{d, t}^{x^n}\}) \\
&\geq t\left(\frac{d}{2} - \frac{1}{2}\left<\rho_{d, t}(\sigma), \frac{1}{n} \sum_{i \in [n]} x^i\right>\right).
\end{align*}
\end{obs}

The ``inverse'' of the above statement is not exactly true, since we have not accounted for the inversions of elements in $\{1, \dots, d + 1, d + t + 1, \dots, 2d + t\}$. Nonetheless, this only contributes at most $2d^2$ to the number of inversions and we can prove the following:

\begin{obs} \label{obs:marginal-to-ranking-inv}
For any $x^1, \dots, x^n, y \in \{\pm 1\}^d$, we have
\begin{align*}
&K(\pi_{d, t}^y, \{\pi_{d, t}^{x^1}, \dots, \pi_{d, t}^{x^n}\}) \\
&\leq t\left(\frac{d}{2} - \frac{1}{2}\left<y, \frac{1}{n} \sum_{i \in [n]} x^i\right>\right) + 2d^2.
\end{align*}
\end{obs}

Ignoring the additive $2d^2$ term, Observations~\ref{obs:marginal-to-ranking-inv} and~\ref{obs:marginal-to-ranking} intuitively tell us that if the 1-way marginal is hard for DP algorithms, then so is rank aggregation. 

\paragraph{Pure-DP} For pure-DP, it is now relatively simple to apply the packing framework of~\citet{HardtT10} for proving a DP lower bound: we can simply pick $x^1 = \dots = x^n$ to be a codeword from an error-correcting code. Observation~\ref{obs:marginal-to-ranking} tells us that this is also a good packing for the Kendall tau metric, which immediately implies Theorem~\ref{thm:lb-pure-dp}.

\paragraph{Approximate-DP} Although there are multiple lower bounds for 1-way marginals in the approximate-DP setting (e.g.,~\citep{BunUV18,SteinkeU15,DworkSSUV15,SteinkeU17}), it does not immediately give a lower bound for the rank aggregation problem because our observations only allow us to recover the \emph{signs} of the marginals, but not their \emph{values}. Fortunately, it is known that signs are already enough to violate privacy~\citep{DworkSSUV15,SteinkeU17} and thus we can reduce from these results\footnote{Another advantage of~\citep{DworkSSUV15,SteinkeU17} is that their the marginal distributions are flexible; indeed, we need distributions which have large standard deviation (i.e., mean is close to $-1$ or $+1$) in order to get a large approximation ratio.}. Another complication comes from the additive $O(d^2)$ term in Observation~\ref{obs:marginal-to-ranking-inv}. However, it turns out that we can overcome this by simply picking $t$ to be sufficiently large so that this additive factor is small when compared to the optimum.


%% file: sm-local.tex
\section{Missing Proofs from Section ``Algorithms in the Local Model''}

\subsection{Proof of \Cref{thm:red-wpone-local}}

To prove \Cref{thm:red-wpone-local}, we need the following algorithm for estimating an aggregated vector, due to~\citet{duchi2014local}.

\begin{lemma}[\cite{duchi2014local}] \label{lem:duchi-vector-sum}
Suppose that each user $i$'s input is a vector $\bv^i \in \R^d$ such that $\|\bv^i\|_2 \leq C$. Then, there exists an $\eps$-local DP algorithm that calculates an estimate $\tbs$ of the average of the input vectors $\bs := \frac{1}{n} \sum_{i \in [n]} \bv^i$ such that
\begin{align*}
\E[\|\bs - \tbs\|_2^2] \leq O\left(\frac{C^2 d}{\eps^2 n}\right).
\end{align*} 
\end{lemma}

\begin{proof}[Proof of Theorem~\ref{thm:red-wpone-local}]
Each user $\ell$'s view of its own ranking $\pi_\ell$ is an $m^2$-dimensional vector $\bv^\ell$ whose $(j, i)$-entry is $\ind[\pi_\ell(j) < \pi_\ell(j)]$; we have $\|\bv^\ell\|_2 \leq m$
for $d = m^2$. Notice that the average of these vectors is exactly $\bw$. We can thus use the $\eps$-local DP algorithm in \Cref{lem:duchi-vector-sum} to compute its estimate $\tbw$. Finally, we run $\cA$ on $\tbw$ and output the ranking found.

The privacy guarantee of the algorithm follows from that of~\Cref{lem:duchi-vector-sum} and the post-processing property of DP.

The $\alpha$-approximation guarantee of $\cA$ implies that
\begin{align} \label{eq:apx-1}
\E_{\sigma}[\barK(\sigma, \tbw)] \leq \alpha \cdot \opt(\tbw).
\end{align}
Furthermore, we may use the Cauchy–Schwarz inequality together with the utility guarantee of \Cref{lem:duchi-vector-sum} to arrive at
\begin{align*}
\E\left[\sum_{i, j \in [m]} |w_{ji} - \tw_{ji}|\right]
&\leq m \cdot \sqrt{\E[\|\bw - \tbw\|_2^2]} \\
(\text{\Cref{lem:duchi-vector-sum}}) &\leq m \cdot \sqrt{O\left(\frac{m^2 \cdot m^2}{\eps^2 n}\right)} \\
&= O\left(\frac{m^3}{\eps \sqrt{n}}\right).
\end{align*}

We may now finish the proof similar to that of \Cref{thm:red-wpone} where $e := O\left(\frac{m^3}{\eps \sqrt{n}}\right)$. By applying Fact~\ref{fact:err-weight-to-err-cost} twice, we then arrive at
\begin{align*}
&\E_{\sigma, \tbw}[\barK(\sigma, \bw)] \\
(\text{Fact}~\ref{fact:err-weight-to-err-cost}) &\leq e + \E_{\sigma, \tbw}[\barK(\sigma, \tbw)] \\
\eqref{eq:apx-1} &\leq m^2e + \alpha \cdot \opt(\tbw) \\
(\text{Fact}~\ref{fact:err-weight-to-err-cost}) &\leq e + \alpha \cdot (\opt(\bw) + e) \\
&= \alpha \cdot \opt(\bw) + (\alpha + 1)e. \qedhere
\end{align*}
\end{proof}

\subsection{Proof of~\Cref{thm:red-expectation-local}}

Before we prove \Cref{thm:red-expectation-local}, we remark that, in the regime of large $q = \omega(m \log m)$, it is possible to improve the error guarantees in \Cref{lem:bassily} and \Cref{thm:red-expectation-local} to $O\left(\frac{\sqrt{q \log m}}{\eps\sqrt{n}}\right)$ and $O\left(\frac{m^2 \sqrt{q \log m}}{\eps\sqrt{n}}\right)$, respectively. This can be done by partitioning the users into $\Theta(q/\log m)$ sets instead of $m$ sets as currently done. However, this improvement does not effect our final additive error bound in \Cref{cor:kwiksort-local} (which only uses $q = O(m \log m)$) and thus we choose to present a simpler version where the number of sets in the partition is fixed to $m$.

\begin{proof}[Proof of~\Cref{thm:red-expectation-local}]
Our algorithm $\cB$ simply works by running $\cA$, and every time $\cA$ queries for a pair $i, j$, we return the answer using the $\eps/2$-DP algorithm $\cQ$ from \Cref{lem:bassily}. If $\cQ$ never returns $\perp$, then we return the output $\sigma^\cA$ of $\cA$. Otherwise, we run the $\eps/2$-DP algorithm from Corollary~\ref{cor:PTAS-local-DP} with $\xi = 1$ and return its output $\sigma^*$. The $\eps$-DP guarantee of the algorithm is immediate from the assumptions and the basic DP composition.

We next analyze the expected Kendall tau distance of the output. Notice that, for the purpose of accuracy analysis, we may view Algorithm~\ref{alg:query-local} as producing $\tw_{ji}$ for all $j, i \in [m]$ (even those not queried); thus we may view the algorithm $\cA$ as running on this $\tbw$ whenever it does not receive $\perp$ from $\cQ$. Let $\mE_{q}$ denote the event that $\cA$ uses more than $q$ queries and let $\mE_{\perp}$ denote the event that $\cQ$ outputs $\perp$ on one of the queries. From our assumption, we have $\Pr[\mE_{q}] \leq \zeta/m^4$; moreover, \Cref{lem:bassily} states that $\Pr[\mE_{\perp} \mid \mE_q] \leq \exp(-0.1q/m) \leq \zeta / m^4$. Combining these, we have that $\Pr[\mE_{\perp}] \leq 2\zeta/m^4$. Furthermore,
\begin{align}
\E_{\sigma^\cA}[\barK(\sigma^\cA, \tbw)] 
&\geq \Pr[\mE_{\perp}] \cdot \E_{\sigma^\cA}[\barK(\sigma^\cA, \tbw) \mid \mE_{\perp}]. \label{eq:expected-cost-conditioned-on-bad-event}
\end{align}
Now, let $\sigma$ denote the output of our algorithm $\cB$. We have
\begin{align*}
&\E_{\sigma, \tbw}[\barK(\sigma, \tbw))] \\
&\leq \Pr[\mE_{\perp}] \cdot \E_{\sigma^\cA, \tw}[\barK(\sigma^\cA, \bw) \mid \mE_\perp] 
+ \Pr[\neg \mE_\perp] \cdot \E_{\sigma^*}[\barK(\sigma^*, \bw)] \\
& \leq \E_{\sigma^\cA}[\barK(\sigma^\cA, \tbw)] 
+ \frac{2\zeta}{m^4} \cdot \left(2 \cdot \opt(\bw) + O_{\alpha}\left(\frac{m^3}{\eps\sqrt{n}}\right) \right) \\
&\leq \alpha \cdot \opt(\bw) + O_{\alpha}\left(m^2 \cdot \frac{q}{\eps\sqrt{nm}}\right) 
+ \frac{2\zeta}{m^4} \cdot \opt(\bw) \\
&\leq (\alpha + \zeta) \cdot \opt(\bw) + O_{\alpha}\left(\frac{m^{1.5} q}{\eps\sqrt{n}}\right),
\end{align*}
where the second inequality follows from~\eqref{eq:expected-cost-conditioned-on-bad-event} and the approximation guarantee of the algorithm from Corollary~\ref{cor:PTAS-pure-DP}, the third inequality follows from similar computations as in the proof of Theorem~\ref{thm:red-wpone} but with the accuracy guarantee from \Cref{lem:bassily} and the last inequality follows because we may assume w.l.o.g. that $m \geq 2$.
\end{proof}

%% file: sm-lb.tex
\section{Missing proofs from Section ``Lower Bounds''}

\subsection{Proofs of Observations~\ref{obs:marginal-to-ranking} and~\ref{obs:marginal-to-ranking-inv}}

\begin{proof}[Proof of Observation~\ref{obs:marginal-to-ranking}]
Consider $K(\sigma, \pi_{d, t}^{x^i})$. Notice that $\rho_{d, t}(\sigma)_j \ne x^i_j$ iff $r_j, r_{j + d}$ occurs in different orders in $r$ compared to $\pi_{d, t}^{x^i})$. When this occurs, it contributes to at least $t$ inversions (of the middle fixed elements). As a result, we can conclude that
\begin{align*}
K(\sigma, \pi_{d, t}^{x^i}) &\geq \sum_{j \in [d]} t \cdot \ind[\rho_{d, t}(\sigma)_j \ne x^i_j] \\
&= \sum_{j \in [d]} t \cdot \frac{1}{2}(1 - \rho_{d, t}(\sigma)_j \cdot x^i_j) \\
&= t\left(\frac{d}{2} - \frac{1}{2}\left<\rho_{d, t}(\sigma), x^i\right>\right).
\end{align*}
Taking the average over all $i \in [n]$ yields the claimed bound.
\end{proof}

\begin{proof}[Proof of Observation~\ref{obs:marginal-to-ranking-inv}]
This follows from observing that the proof of Observation~\ref{obs:marginal-to-ranking} already exactly counts all inversions except those between the elements $\{1, \dots, d + 1, d + t + 1, \dots, 2d + t\}$, and the latter contributes at most $2d^2$ inversions.
\end{proof}

\subsection{Lower Bound for Pure-DP in the Central Model}

The pure DP case (in the central model) is the easiest to handle, as we can simply apply a packing argument with $x^1 = \dots = x^n$ being a codeword in an error-correcting code, which already circumvents the two ``issues'' discussed above.

\begin{proof}[Proof of Theorem~\ref{thm:lb-pure-dp}]
We assume w.l.o.g. that $m$ is divisible by 3, and let $d = t = m/3$. Let $a_1, \dots, a_T \in \{-1, +1\}^d$ denote an error-correcting code with distance at least $0.49d$; it is known that there exists such a code with $T = \exp(\Omega(d))$.

Now, let $\cA$ be any $\eps$-DP $(\alpha, 0.01 m^2)$-approximation algorithm for the rank aggregation problem. Let $D^i$ denote the dataset that includes $n$ copies of of $\pi_{d, t}^{a_i}$. Notice that $\opt(D^i) = 0$; thus we have $\E_{\sigma \sim \cA(D^i)}[\barK(\sigma, D^i)] \leq 0.01 m^2$. Define $S^i := \{\sigma \mid \barK(\sigma, D^i) \leq 0.02m^2\}$ From Markov's inequality, we have $\Pr_{\sigma \sim \cA(D^i)}[\sigma \in S^i] \geq 0.5$. Now, consider any distinct $i, i' \in [T]$ and $\sigma \in \mathbb{S}_{m}$, we have
\begin{align*}
&\barK(\sigma, D^i) + \barK(\sigma, D^{i'}) \\
(\text{Observation}~\ref{obs:marginal-to-ranking}) &\geq d\left(d + \frac{1}{2} \left<\rho_{d, t}(\sigma), a^i + a^{i'}\right>\right) \\
&\geq d\left(d - \frac{1}{2} \|\rho_{d, t}(\sigma)\|_{\infty} \|a^i + a^{i'}\|_1\right) \\
&= d\left(d - \frac{1}{2} \|a^i + a^{i'}\|_1\right).
\end{align*} 
Now, from the distance guarantee of the error-correcting code, we have $\|a^i + a^{i'}\|_1 \leq 2(0.51d)$. Plugging this back into the above, we have
\begin{align*}
\barK(\sigma, D^i) + \barK(\sigma, D^{i'}) &\geq 0.49d^2 > 0.04m^2.
\end{align*}
This means that $S^i, S^{i'}$ are disjoint. As such, we have
\begin{align*}
1 &\geq \sum_{i \in [T]} \Pr_{\sigma \sim \cA(\emptyset)}[\sigma \in S^i] 
\geq \sum_{i \in [T]} e^{-\eps n}\Pr_{\sigma \sim \cA(D^i)}[\sigma \in S^i] 
\geq \sum_{i \in [T]} e^{-\eps n} \cdot 0.5 
= 0.5 T \cdot e^{-\eps n},
\end{align*}
where the second inequality follows from the assumption that $\cA$ is $\eps$-DP. The above inequality, together with $T \geq \exp(\Omega(d))$, implies that $n \geq \Omega(d / \eps) = \Omega(m/\eps)$ as desired.
\end{proof}

\subsection{Lower Bound for Approximate-DP in the Central Model}

As stated earlier, we will reduce from the lower bound for the 1-way marginal problem from~\citep{SteinkeU17}. To state their result, we use $\Beta(\alpha, \beta)$ to denote the beta distribution with parameter $\alpha, \beta$. For notational convenience, we assume that the Bernoulli distribution $\Ber(q)$ is over $\{-1, +1\}$ instead of $\{0, 1\}$. For a vector $\bq \in [0, 1]^d$, we use $\Ber(\bq)$ as a shorthand for $\Ber(q_1) \times \cdots \times \Ber(q_d)$.

\begin{theorem}[{\cite[Theorem 3]{SteinkeU17}}] \label{thm:fingerprinting}
Let $\beta, \gamma > 0$ and $n, d \in \N$. Let $\cB$ be any $(1, \gamma/n)$-DP algorithm whose output belongs to $\{-1, +1\}^d$. Let $q_1, \dots, q_d$ be drawn i.i.d. from $\Beta(\beta, \beta)$ and $x^1, \dots, x^n$ be i.i.d. drawn from $\Ber({\bq})$. If
\begin{align} \label{eq:fingerprinting-correlation-upper-bound}
\E_{\bq, x^1, \dots, x^n, \bw \sim  \cA(x^1, \dots, x^n)}\left[\sum_{j \in [d]} w_j \cdot \left(q_j - 0.5\right)\right] \geq \gamma d,
\end{align}
then $n \geq \gamma\beta \sqrt{d}$.
\end{theorem}

We will now prove our lower bound (\Cref{thm:lb-apx-DP-all-eps}). In fact, it suffices to only prove the statement for $\eps = 1$, from which, using the standard group privacy-based techniques (see e.g.,~\citep{SteinkeU16} for more details), we can get the lower bound in Theorem~\ref{thm:lb-apx-DP-all-eps} for $\eps \leq 1$ where the number of required users grow with $1/\eps$.

\begin{lemma} \label{lem:lb-apx-DP}
For any constant $\alpha > 0$, there exists $c > 0$ (depending on $\alpha$) such that there is no $(1, o(1/n))$-DP $(\alpha, c m^2)$-approximation algorithm for the rank aggregation problem for $n = o(\sqrt{m})$.
\end{lemma}

\begin{proof}
We assume w.l.o.g. that $\alpha$ is an integer and that $n \geq 25$.
Let $\beta > 0$ be sufficiently small so that the following holds:
\begin{align*}
\E_{q \sim \Beta(\beta, \beta)}[|2q - 1|] > 1 - 0.5/\alpha.
\end{align*}
Note that this exists because as $\beta \to 0$, the left hand side term converges to 1. 

Let $\gamma = 0.05$ and $c = 0.00001\alpha$. Note that $\beta, \gamma, c$ are constants depending only on $\alpha$ (but not $n$). Furthermore, we may assume that $n > 25$.

Now, consider any $(1, \gamma/n)$-DP $(\alpha, cm^2)$-approximation algorithm $\cA$ for the rank aggregation problem. Again, assume w.l.o.g. that $m$ is divisible by $40\alpha + 2$ and let $d = m/(40\alpha + 2)$ and $t = 40\alpha d$. Now, let $q^1, \dots, q^d, x^1, \dots, x^n$ be sampled as in the statement of Theorem~\ref{thm:fingerprinting}. For notational brevity, we will simply write $\pi^x$ and $\rho$ instead of $\pi^x_{d, t}$ and $\rho_{d, t}$. Furthermore, let $s \in \{-1, +1\}^d$ be such that $s_j = (-1)^{\ind[q^j > 1/2]}$ for all $j \in [d]$.

From the guarantee of $\cA$, we have
\begin{align*}
&\E_{\sigma \sim \cA(\pi^{x^1}, \dots, \pi^{x^n})}[\barK(\sigma, \{\pi^{x^1}, \dots, \pi^{x^n}\})] \\
&\leq \alpha \cdot \opt(\{\pi^{x^1}, \dots, \pi^{x^n}\}) + c m^2 \\
&\leq \alpha \cdot \barK(\pi^s, \{\pi^{x^1}, \dots, \pi^{x^n}\}) + cm^2 \\
(\text{Observation}~\ref{obs:marginal-to-ranking-inv})
&\leq \alpha \left(t\left(\frac{d}{2} - \frac{1}{2}\left<s, \frac{1}{n} \sum_{i \in [n]} x^i\right>\right) + 2d^2\right) + cm^2 \\
&= \alpha \left(t\left(\frac{d}{2} - \frac{1}{2}\sum_{j \in [d]} (-1)^{\ind[q^j < 1/2]} \cdot \frac{1}{n} \sum_{i \in [n]} x^i_j
\right)\right) +(2\alpha d^2 + cm^2).
\end{align*}

Now, let the $\cB$ denote the algorithm that takes in $x^1, \dots, x^n \in \{-1, +1\}^d$, runs $\cA$ to get $\sigma$ and then output $w = -\rho(\sigma)$. Using Observation~\ref{obs:marginal-to-ranking} together with the above inequality, we arrive at
\begin{align*}
&\E_w\left[t\left(\frac{d}{2} - \frac{1}{2}\left<w, \frac{1}{n} \sum_{i \in [n]} x^i\right>\right)\right] \\
&\leq \alpha \left(t\left(\frac{d}{2} - \frac{1}{2}\left<s, \frac{1}{n} \sum_{i \in [n]} x^i\right>\right) + 2d^2\right) + cm^2 \\
&= \alpha \left(t\left(\frac{d}{2} - \frac{1}{2}\sum_{j \in [d]} (-1)^{\ind[q^j < 1/2]} \cdot \frac{1}{n} \sum_{i \in [n]} x^i_j
\right)\right) + (2\alpha d^2 + cm^2).
\end{align*}
Rearranging, we arrive at
\begin{align}
&\E_w\left[\left<w, \frac{1}{n} \sum_{i \in [n]} x^i\right>\right] 
\geq d - \alpha \cdot \left(d - \sum_{j \in [d]} (-1)^{\ind[q^j < 1/2]} \cdot \frac{1}{n} \sum_{i \in [n]} x^i_j\right) - \frac{4\alpha d^2 + 2cm^2}{t}. \label{eq:xpand-emp-prod}
\end{align}

We will now take expectation on both sides w.r.t. $x^1, \dots, x^n$ sampled i.i.d. from $\Ber(\bq)$. On the right hand side, we get
\begin{align}
&\E_{x^1, \dots, x^n} \left[d - \alpha \left(d - \sum_{j \in [d]} (-1)^{\ind[q^j < 1/2]} \cdot \frac{1}{n} \sum_{i \in [n]} x^i_j\right) - \frac{4\alpha d^2 + 2cm^2}{t}\right] \nonumber \\
&=  d - \alpha \left(d - \sum_{j \in [d]} (-1)^{\ind[q^j < 1/2]} \cdot (2q_j - 1)\right) - \frac{4\alpha d^2 + 2cm^2}{t} \nonumber\\
&=  d - \alpha \left(d - \sum_{j \in [d]} |2q_j - 1|\right) - \frac{4\alpha d^2 + 2cm^2}{t}. \label{eq:exp-bern-rhs}
\end{align}

For the left hand side, notice that 
\begin{align} \label{eq:bern-deviation}
\E\left[\left|\frac{1}{n} \sum_{i \in [n]} x^i_j - (2q_j - 1)\right|\right] \leq \frac{1}{\sqrt{n}}.
\end{align}
From this, we have
\begin{align}
& \E_{x^1, \dots, x^n, w}\left[\sum_{j \in [d]} w_j \cdot \left(q_j - 0.5\right)\right] 
= \frac{1}{2} \E_{x^1, \dots, x^n, w}\left[\sum_{j \in [d]} w_j \cdot \left(2q_j - 1\right)\right] \nonumber \\
&= \frac{1}{2} \E_{x^1, \dots, x^n, w}\left[\left<w, \frac{1}{n} \sum_{i \in [n]} x^i\right>\right] 
- \frac{1}{2}\E_{x^1, \dots, x^n, w}\left[\sum_{j \in [d]} w_j \cdot \left(\frac{1}{n} \sum_{i \in [n]} x^i_j - (2q_j - 1)\right)\right] \nonumber \\
&\geq \frac{1}{2} \E_{x^1, \dots, x^n, w}\left[\left<w, \frac{1}{n} \sum_{i \in [n]} x^i\right>\right] 
- \frac{1}{2}\E_{x^1, \dots, x^n, w}\left[\sum_{j \in [d]} \left|\frac{1}{n} \sum_{i \in [n]} x^i_j - (2q_j - 1)\right|\right] \nonumber \\
&\geq \frac{1}{2} \E_{x^1, \dots, x^n, w}\left[\left<w, \frac{1}{n} \sum_{i \in [n]} x^i\right>\right] - \frac{d}{2\sqrt{n}}, \label{eq:prod-from-emp-prod}
\end{align} 
where the first inequality follows from $w_j \in \{-1, +1\}$ and the second one follows from~\eqref{eq:bern-deviation}.

Combining~\eqref{eq:xpand-emp-prod},~\eqref{eq:exp-bern-rhs} and~\eqref{eq:prod-from-emp-prod}, we arrive at
\begin{align*}
\E_{x^1, \dots, x^n, w}\left[\sum_{j \in [d]} w_j \cdot \left(q_j - 0.5\right)\right] 
\geq \frac{1}{2}\left(d - \alpha \left(d - \sum_{j \in [d]} |2q_j - 1|\right) - \frac{4\alpha d^2 + 2cm^2}{t}\right) - \frac{d}{2\sqrt{n}}.
\end{align*}
Taking the expectation over $q^1, \dots, q^d \sim \Beta(\beta, \beta)$ on both side and using our choice of $\beta$, we arrive at
\begin{align*}
&\E_{\bq, x^1, \dots, x^n, w}\left[\sum_{j \in [d]} w_j \cdot \left(q_j - 0.5\right)\right] 
\geq \frac{1}{2}\left(d - \alpha\left(d - (1 - 0.5/\alpha)d\right) - \frac{4\alpha d^2 + 2cm^2}{t}\right) - \frac{d}{2\sqrt{n}} \\
&= \frac{1}{2}\left(0.5 d - \frac{4\alpha d^2 + 2cm^2}{t}\right) - \frac{d}{2\sqrt{n}} 
\quad \geq \quad \frac{1}{2}\left(0.5 d - 0.1 d - 0.1 d\right) - 0.1 d 
\quad \geq \quad 0.05 d = \gamma d,
\end{align*}
where the second-to-last inequality follows from our choices of $t, c$ and our assumption that $n \geq 25$.

As a result, we can conclude via Theorem~\ref{thm:fingerprinting} that $n \geq \Omega(\sqrt{d})$ as desired.
\end{proof}

